\documentclass[12pt]{article}

\usepackage{amsmath,amssymb,amsthm,amstext}

\usepackage{color}
\usepackage{graphicx,epsfig}
\usepackage{fullpage}

\newcommand\es{\Lambda}
\newcommand{\pair}[1]{[ #1 ]}

\renewcommand{\le}{\leqslant}
\renewcommand{\ge}{\geqslant}

\newcommand{\eps}{\varepsilon}
\newtheorem{theorem}{Theorem}

\newtheorem{lemma}{Lemma}
\theoremstyle{remark}
\newtheorem{definition}{Definition}
\newtheorem{remark}{Remark}

\title{Information disclosure 
in the framework of Kolmogorov complexity}
\author{Nikolay Vereshchagin\thanks{The results presented in Sections 4 are supported by Russian Science Foundation (20-11-20203). The results presented in Sections 3 have been supported
    by the Interdisciplinary Scientific and Educational School of Moscow University ``Brain, Cognitive Systems, Artificial Intelligence''.}\\
Moscow State University and
HSE University, Russian Federation.
}
\date{}
\begin{document}

\maketitle
\begin{abstract}
We consider the network consisting of three nodes 1, 2, 3 connected by two open channels $1\to 2$ and $1\to 3$. The information present in the node 1 consists of four strings $x,y,z,w$. The nodes  2, 3 know $x,w$ and need to know  $y,z$, respectively. We want to arrange transmission of information over the channels so that both nodes 2 and 3 learn what they need and the disclosure of information is as small as possible. By  information disclosure we mean the amount of information in the strings transmitted through channels about $x,y,z,w$ (or about $x,w$). We are also interested in whether it is possible to minimize the disclosure of information and simultaneously minimize the length of words transferred through the channels.  
\end{abstract}  

\section{Introduction}\label{s1}

Assume that a finite directed graph is given, its arcs  are called \emph{channels}.
Assume further each node is assigned two strings, called the \emph{input string} and the \emph{output string} of that node.
 The input string is understood as the information that is present in the node, and
 the  output string as the information that is needed in the node. 
The goal is to arrange the information transmission so that
every node gets what it needs. By the information transmission we mean an assignment
of strings to channels in such a way that  for every node its output string has low 
Kolmogorov complexity relative to the tuple consisting of its input string and all the strings assigned to
its input channels. ``Low'' means complexity of order
$o(n)$, where $n$ is the maximal length of all input and output strings.

Complexity of information transmission is measured by the following quantities.
First, we consider the amount of transmitted information, that is,  Kolmogorov complexity of   strings assigned to channels.
The less that complexity is the better. In this setting the problem was studied (for different networks) in the papers~\cite{shen, bglvz,muchnik, rom}.
Second, we consider information disclosure.
That is, we assume that an eavesdropper gets the tuple $T$ consisting of all the strings assigned to channels
and measure how much information has $T$ about input and output strings.

We know two papers,  \cite{much} and \cite{romzim}, devoted to information disclosure in the framework of Kolmogorov complexity.
The authors of~\cite{romzim}  considered two-parties
multi-round protocols of a random key generation 
in such a way that that the eavesdropper has no information about the generated key. Before  communication each party
has her/his own string.
The authors of~\cite{romzim} showed that
the maximal length of generated key is equal to the mutual information in those strings.

In~\cite{much}, the following problem was studied: Alice has strings
$x,y$ and Bob has only $x$ and wants to know $y$;  to this end Alice is allowed to send to Bob a string $p$; they both want the
eavesdropper to obtain as little information about $y$ as possible (see Fig.~\ref{f7}(b)).  It is not hard to show 
that any such string  $p$ (with negligible $C(y|p,x)$) has at least 
$C(y)-C(x)$ bits of information about $y$ and its complexity is at least $C(y|x)$.
It is shown in~\cite{much} that both lower bounds can be attained simultaneously, that is, there is 
a string $p$ of length $C(y|x)+O(\log n)$, that has  $\max\{C(y)-C(x),0\}+O(\log n)$ information about $y$ (where $n$ denotes the maximum of lengths of
$x,y$).

In this paper, we study the network shown on Fig.~\ref{f7}(a),
as well as networks obtained from it by identifying some input or output strings. Those networks are shown on Fig.~\ref{f7}(b), (c), (d), (e), (f).
\begin{figure}[t]
\begin{center}
\includegraphics{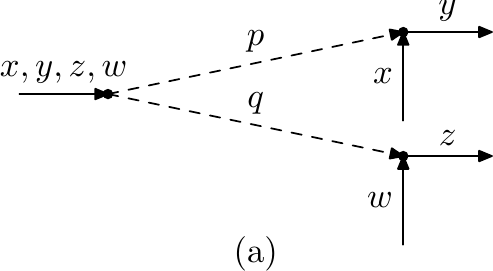}\qquad
\includegraphics{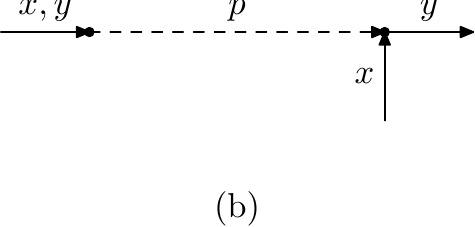}\qquad
\includegraphics{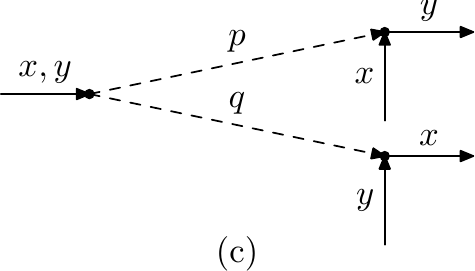}
\end{center}

\begin{center}
\includegraphics{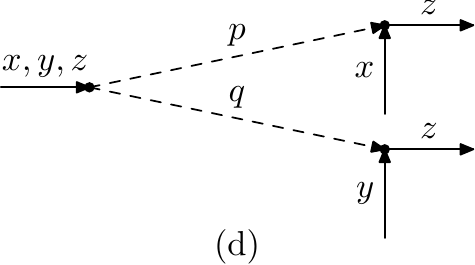}\qquad
\includegraphics{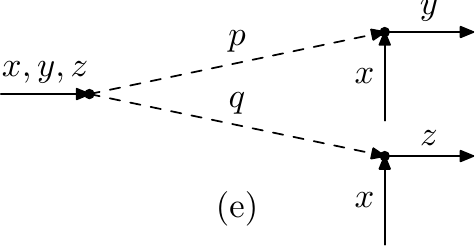}\qquad\includegraphics{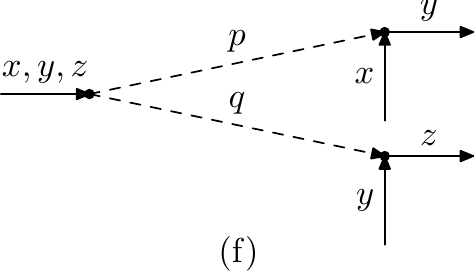}
\end{center}
\caption{The networks studied in this paper}\label{f7}
\end{figure}
Channels are shown by dashed lines. The left node is understood as the node transmitting information to two consumers (the nodes to the right).
Its input string is the tuple consisting of strings  $w,x,y,z$  and the output string is empty.
The input string of the right upper node is  $x$ and its output string is
 $y$. Similarly, the input string of the right lower node is  $w$ and its output string is
 $z$. To arrange information transmission in the network of Fig.~\ref{f7}(a), we have to find strings 
$p,q$ such that all quantities $C(p,q|x,y,z), C(y|p,x), C(z|q,w)$ are negligible.
For instance, we may let $p=y$ and $q=z$. More specifically, we say that a pair of strings
\emph{$p,q$ is $\eps$ information transmission} in this network if all quantities 
$$
C(p|w,x,y,z),
C(q|w,x,y,z), C(y|p,x), C(z|q,w)
$$ 
are less than $\eps$.

Now we will define information disclosure for this network. 
Generally, we may define information disclosure as the amount of information 
in the pair $(p,q)$ about any non-empty group of
input/output strings. In this way we can obtain 15 different notions.
To simplify things, we will consider only two groups, the group $\{x,y,z,w\}$ consisting of all input and output strings,
and the group  $\{x,w\}$ consists of all input strings that are not output strings. Strings $x,w$ constitute private information of right nodes
and may be considered more secret than strings $y,z$.
The information in $(p,q)$ about the tuple $(x,y,z,w)$ is called the 	\emph{total disclosure}
and the information in $(p,q)$ about  $(z,w)$ is called the 	\emph{private disclosure}.

Thus we will consider the following four parameters:
\begin{itemize}
\item the amounts of transmitted information, $C(p),C(q)$,
\item the total disclosure $I(p,q:w,x,y,z)$, and 
\item the private disclosure  $I(p,q:w,x)$.
\end{itemize}
We will denote the maximum of lengths of strings $x,y,z,w$ by $n$ .

It is not hard to see that the total  disclosure is equal to $C(p,q)$ with accuracy $2\eps+O(\log n)$. 
Indeed, by Commutativity of information we have
$$
I(p,q:w,x,y,z)=C(p,q)-C(p,q|w,x,y,z)+O(\log n),
$$ 
hence  
$$
C(p,q)-2\eps+O(\log n)\le I(p,q:w,x,y,z)\le C(p,q)+O(\log n).
$$ 
Besides, all lower bounds for 
$C(p,q)$ that we are able to prove hold for $I(p,q:w,x,y,z)$ as well (with accuracy  $O(\log n)$). 
On the other hand, all upper bounds for the total disclosure hold for $C(p,q)$ as well. 
Therefore, in the present paper the difference between $C(p,q)$ and $I(p,q:w,x,y,z)$
is not important, and in the sequel we will consider $C(p,q)$ instead of the total disclosure. 

One can verify that, excluding trivial cases, identification of some input or output strings in the network shown on Fig.~\ref{f7}(a) 
produces one of the networks Fig.~\ref{f7}(b,c,d,e,f).
For each of the networks shown on Fig.~\ref{f7}, we are interested in the following questions.
For given strings $w,x,y,z$, how low can be the quantities
$C(p),C(q),C(p,q)$,  
and the private disclosure for $\eps$ information transmission pairs $p,q$?
The  private disclosure is defined as $I(p,q:x,w)$ for the network of Fig.~\ref{f7}(b), 
as $I(p:x)$ for the network Fig.~\ref{f7}(b),
as $I(p,q:x,y)$ for the network Fig.~\ref{f7}(d), and
as $I(p,q:x)$ for the networks Fig.~\ref{f7}(e,f).
And the  private disclosure is 0 by definition for the network Fig.~\ref{f7}(c),
as it has no private inputs.
The second question is: 
can the minimal possible values of those quantities attained simultaneously?

In brief, the results of this paper are the following:
\begin{itemize}
\item We show that for the network of Fig.~\ref{f7}(a) some results from the paper~\cite{vermuch}
imply that the minimal possible $C(p,q)$ for 
$\eps$ information transmission cannot be expressed with accuracy $o(n+\eps)$ through complexities $w,x,y,z$, their pairs, triples and the quadruple
$(w,x,y,z)$.
The same holds for the private disclosure.
\item
For the networks of Fig.~\ref{f7}(b,c,d),
using some known results, we find the minimum values of all the four quantities and show that they all can be attained simultaneously.
\item
For the networks of Fig.~\ref{f7}(e,f), we find minimal possible values of 
all the four quantities and show that they can \emph{not} be attained simultaneously,
neither for the quantities  
 $C(p),C(q),C(p,q)$, nor for the quantities
$C(p),C(q),I(p,q:x)$. These are the main results of the paper.
\end{itemize}

In the next section we provide the main definitions and results on Kolmogorov complexity.
Section~\ref{s3} contains the definitions related to information transmission in networks and previous results.
Section~\ref{s4} contains the main results and their proofs.

\section{Preliminaries}\label{s2}

We write \emph{string} to denote a  finite  binary  string.  Other
finite objects, such  as  pairs of  strings,  may be encoded  into
strings  in  natural ways. The set  of  all strings is denoted  by
$\{0,1\}^*$ and the length of a string $x$ is denoted by $|x|$. The empty string is denoted by $\es$.

Let $\bar x$ denote the string
$$
\underbrace{000\dots0}_{|x|\text{ times}}1x=0^{|x|}1x.
$$
We shall use the string $\bar xy$ to encode the pair
$(x,y)$ of strings; the notation $[x,y]$ will mean  the
same as $\bar xy$. 
Let $\log n$ denote the binary logarithm of $n$.

A \emph{programming language} is a partial computable function
$F$ from $\{0,1\}^*\times\{0,1\}^*$ to $\{0,1\}^*$.
The first argument of $F$ is called a program,
the second argument is called the input, and
$F(p,x)$ is called the output of program $p$ on input $x$.
A programming language $U$ is called \emph{optimal}
if for any other
programming language $F$ there exists a string $t_F$
such that $U(t_Fp,x)=F(p,x)$ for  all $p,x$.
By Solomonoff -- Kolmogorov theorem (see e.g. \cite{lv,SUV})
optimal programming languages
exist. We fix some optimal programming language $U$ and define
\begin{itemize}
\item $C(x|y)=\min\{|p|\mid U(p,y)=x\}$ (conditional Kolmogorov
complexity
of $x$ relative to $y$),
\item $C(x)=C(x|\es)$ (Kolmogorov complexity
of $x$),
\item  $I(x:y)=C(y)-C(y|x)$ (information in $x$ about $y$),
\item  $I(x:y|z)=C(y|z)-C(y|[x,z])$ (information in $x$ about $y$ relative to $z$),
\item $J(x:y)=C(x)+C(y)-C(\pair{x,y})$, $J(x:y|z)=C(x|z)+C(y|z)-C(\pair{x,y}|z)$ (mutual information
between  $x$ and $y$).
\item We say that strings $x$ and  $y$
are \emph{independent} if $J(x:y)$ is close to 0.
\end{itemize}
If $U(p,y)=x$, we say that $p$ is \emph{a  program for $x$ relative to $y$}.
If $U(p,\es)=x$, we say that \emph{$p$ is a program for $x$}.


Instead of  $C(\pair{x,y})$, $I(x:[y,z])$, $C(z|\pair{x,y})$, and  $C(\pair{x,y}|z)$ we shall
write  $C(x,y)$,  $I(x:y,z)$, $C(z|x,y)$, 
and  $C(x,y|z)$ respectively.
We use the following well known facts
(see~\cite{lv,SUV}):
\begin{itemize}
\item $C(x)\le |x|+O(1)$;
\item $C(x|y)\le C(x)+O(1)$;
\item $C(x|y,z)\le C(x|z)+O(1)$;
\item
for any partial computable function $f(x)$ there is a constant $c$ such that
$C(f(x)) \le C(x)+c$ for all $x$ in the domain of $f$;
\item \textbf{Chain rules}:
\begin{align*}
C(x,y)& = C(x)+C(y|x)+O(\log C(x,y)),\\
C(x,y|z)& = C(x|z)+C(y|x,z)+O(\log C(x,y|z)),\\
C(x,y:u|z)& = C(x:u|z)+C(y:u|x,z)+O(\log C(x,y,u|z))
\end{align*}
for all $x,y,z,u$;
\item \textbf{Commutativity of information}:
$$
I(x:y)=J(x:y)+O(\log C(x,y))=I(y:x)+O(\log C(x,y)),
$$
\item \emph{the upper graph} $\{(x,y,i)\mid C(x|y)\le i\}$ of conditional Kolmogorov complexity
is computably enumerable.
\end{itemize}

\section{Information transmission in networks: definitions and previous results}\label{s3}

\subsection{Main definitions}

\begin{definition}
We say that a pair of strings 
\emph{$(p,q)$ is an $\eps$ information transmission} in the network of Fig.~\ref{f7}(a) if all the quantities 
$$
C(p|w,x,y,z),
C(q|w,x,y,z), C(y|p,x), C(z|q,w)
$$ 
are less than $\eps$.
In a similar way we define $\eps$ information transmission in other networks.
\end{definition}

\begin{definition}
Let $f(x,y,z,w)$ and $g(p,q,x,y,w,z)$ be some integer valued functions.
We say that $f$ is an \emph{upper bound for $g$} in the network of Fig.~\ref{f7}(a)
if for all strings $x,y,z,w$ of length at most $n$ there is a pair of strings $p,q$ of length $O(n)$ that is
$o(n)$ information transmission in the network and $g(p,q,x,y,w,z)\le f(x,y,z,w)+o(n)$.
(For example,  $C(y|x)$
is an upper bound for  $C(p)$ in the network of Fig.~\ref{f7}(a).)

We say that  $f$ is 
a \emph{lower bound for  $g$} in the network of Fig.~\ref{f7}(a)
if 
for all strings    $x,y,z,w$ and for all  $\eps$ information transmission pairs $(p,q)$,
it holds 
$$g(p,q,x,y,w,z)\ge f(x,y,z,w)-c\eps-d\log n-e,$$
where $c,d$ are some absolute constants,
$n$ is the maximal length of strings 
$x,y,z,w,p,q$,
 and the constant $e$
depends on the choice of the optimal programming language in the definition of Kolmogorov complexity.
(For instance,  $C(y|x)$
is a lower bound for $C(p)$ in the network of Fig.~\ref{f7}(a).)

If a function $f$ is both an upper and lower bound 
for  $g$,
then we say that \emph{the minimal possible value of 
$g$  is} 
$f$. (For example, the minimal possible complexity of $p$
in the network of Fig.~\ref{f7}(a) is $C(y|x)$.)
In a similar way we define the notions of upper and lower bounds for other networks.
\end{definition}

\begin{definition}
\emph{The profile of a quadruple of strings $x,y,z,w$} 
is defined as the tuple consisting of 15 numbers
$$
C(x),C(y),C(z),C(w), C(x,y),C(x,z),C(x,w), \dots, C(x,y,z,w) 
$$ 
We call a function of  $x,y,z,w$, a \emph{profile function} if it is a function of the profile of $x,y,z,w$.
If there is no profile function $f$ of $x,y,z,w$ that is the minimal possible value 
of  $g$ for a network, then we say that \emph{the minimal possible value of  
$g$  for that network is not a profile function}. 
 \end{definition}

\subsection{The network of Fig.~\ref{f7}(a)} 

The minimal possible values of $C(p),C(q)$ in all the networks of Fig.~\ref{f7} can be found quite easily.
Namely,  the minimal possible  $C(p)$ or $C(q)$ is equal to Kolmogorov complexity of the output string
relative to the input string in the node where the corresponding arrow directs to. 
For example, in the network of Fig.~\ref{f7}(a), the minimal possible $C(p)$ is equal to  $C(y|x)$ and the minimal $C(q)$ is 
$C(z|w)$. Both minimal values can be attained simultaneously.

More specifically, if a pair $(p,q)$ is $\eps$ information transmission, then
\begin{align}
&C(p)\ge C(y|x)-\eps-O(\log n),\label{eq13}\\
&C(q)\ge C(z|w)-\eps-O(\log n).\label{eq14}
\end{align}
This can be proved via cuts. The cut technique works as follows.
\begin{quote}
We choose a set of nodes, called  \emph{a cut}. Let  
$A$ denote the tuple consisting of all input strings of the nodes of the cut and let
$B$  denote the tuple consisting of all output strings of the nodes of the cut. Finally, let $P$
denote the tuple consisting of labels of channels directed to the nodes of the cut. 
Then $C(B|A)$ is a lower bound for $C(P)$ and even for $I(P:A,B)$ (hence for the total disclosure).
\end{quote}
More specifically, we have the following
\begin{lemma}\label{l3}
For all strings $A,B,P$  of length at most  $n$ we have 
\begin{align}\label{eq33}
&C(B|A)\le I(P:(A,B))+C(B|P,A)+O(\log n).
\end{align}
\end{lemma}
 \begin{proof}
Move $C(B|P,A)$ to the left hand side of the inequality~\eqref{eq33}.
Then the left hand side becomes
$ I(P:B|A)$. 
By the Chain rule the right hand side equals 
 $I(P:A)+I(P:B|A)$ (ignoring  $O(\log n)$ terms), and hence is larger than the left hand side.
\end{proof}

Since $C(B|P,A)=O(\eps)$, Lemma~\ref{l3} implies that  $C(B|A)$ is indeed a lower bound for $I(P:A,B)$.
Plugging in Lemma~\ref{l3}
$A=x,B=y,P=p$ and $A=w,B=z,P=q$, we get the inequalities  \eqref{eq13} and \eqref{eq14}, respectively.

For some networks, cuts yield even better lower bounds for $I(P:(A,B))$ than $C(B|A)$.
More specifically, sometimes we can remove some input strings from 
$A$ keeping $C(B|P,A)=O(\eps)$. 
This happens when some string belongs both to $A$ and $B$. 
For example, consider the cut consisting of both right nodes in the network shown on Fig.~\ref{f7}(f). Then we may exclude 
the string $y$  from $A$, since it can be obtained from $p$ and $x$ and some extra $\eps$ bits. Hence we have  
$C(y,z|p,q,x)<2\eps+O(1)$ and thus $C(y,z|x)$ is a lower bound for the total disclosure in that network.
Another example: consider the cut consisting of both right nodes in the network of Fig.~\ref{f7}(c). Then we can exclude $x$ or $y$ from  $A$ (but cannot
exclude both $x,y$). In this way we can show
that both  $C(x,y|x)$ and $C(x,y|y)$ are lower bounds for the total disclosure in that network.

Obviously, $C(y|x)$ and $C(z|w)$ are also upper bounds for  $C(p)$ and $C(q)$, respectively 
(for the network of Fig.~\ref{f7}(a), and similar upper bounds hold for other networks from Fig.~\ref{f7}). 
Namely, we can let $p$ be the minimum length program for $y$ relative to $x$ (its length equals $C(y|x)$). 
Similarly, we can let $q$ be the minimum length program for  $z$ relative to $w$.

The minimal  $C(p,q)$ for the network Fig.~\ref{f7}(a) was studied in the paper~\cite{vermuch}. It was shown that
the minimal  $C(p,q)$ for this network is not a profile function of $x,y,z,w$. Hence
the minimal total disclosure is not a profile function of  $x,y,z,w$. 

More specifically, the following function 
$$
f(x,y,z,w)=\min\{C(p,q)\mid  y=U(p,x), z=U(q,w)\} 
$$
was studied  in~\cite{vermuch}.
Here $U$ is the optimal programming language.
Notice that $f(x,y,z,w)$ is different from
the minimal complexity of a pair that is $\eps$ information transmission in this network.
The difference is in that in the definition of $f(x,y,z,w)$ we do not require  $C(p|x,y,z,w),C(q|x,y,z,w)$ be negligible.
Another difference is that we replace the requirements of negligibility of $C(y|p,x),C(z|q,y)$ by stronger requirements
$y=U(p,x)$, $z=U(q,w)$. However, this second difference is not important, since the inequalities  $C(y|p,x),C(z|q,y)<\eps$ imply
that adding to  $p,q$ some extra 
$\eps$ bits we can obtain $p',q'$ with $y=U(p',x)$, $z=U(q',w)$. 
However, the first difference is important.

In~\cite{vermuch}, for every natural number
$n$ two quadruples of strings
$x,y,z,w$ and $x',y',z',w'$ of length $O(n)$ were defined.
The profiles of those quadruples coincide with 
accuracy $O(\log n)$ in all components and
$f(x,y,z,w)\le n$ but $f(x',y',z',w')\ge 1.5n$ (with accuracy $O(\log n)$).
Besides, the profiles of those tuples has the following feature:
$C(y,z|x,w)=O(\log n)$. Hence for all  $p,q$ the total disclosure equals the private disclosure with accuracy $O(\log n)$
(for both quadruples).  The second feature is that the upper bound for $C(p,q)$ for the first quadruple 
is attained for a pair $p,q$ with  $C(p,q|x,y,z,w)=O(1)$. Therefore, for the first quadruple both the total and private disclosures are at most
 $n$ while for the second quadruple both are at least  $1.5n$ (with accuracy $O(\log n)$).
Hence both minimal total and private disclosures are not profile functions for this 
network.


 \subsection{The network Fig.~\ref{f7}(c) and information distance} 

We skip the network of Fig.~\ref{f7}(b) as an obvious one and proceed to the network of Fig.~\ref{f7}(c).
Private disclosure is zero for this network, independently of the choice of $p,q$, since there are no private inputs. 

The minimal $C(p,q)$ for this network was essentially found 
in the paper~\cite{bglvz}. It was called  
  \emph{the information distance between $x$ and $y$}.
More precisely, the information distance between $x$ and $y$ is defined as $\min \{C(p,q)\mid U(p,x)=y, U(q,y)=x\}$,
where $U$ stands for the optimal programming language.
In~\cite{bglvz}, it was proved that the information distance between $x$ and $y$ equals  $\max\{C(y|x),C(x|y)\}$ 
with accuracy $O(\log\max\{C(y|x),C(x|y)\})$.
It follows from the proof, that the pair $(p,q)$  witnessing this equality has logarithmic complexity relative to
the pair  $x,y$, hence  $\max\{C(y|x),C(x|y)\}$
is an upper bound for $C(p,q)$ for this network.
On the other hand, $\max\{C(y|x),C(x|y)\}$ is a lower bound as well, which can be shown using
the cut consisting of both right nodes (or with two cuts, consisting of one right nodes each). 

The results of~\cite{bglvz} leave open the following question,
can the minimum values of $$C(p),C(q),C(p,q)$$ be attained simultaneously?  The positive answer 
to this question follows from An. Muchnik's theorem from the next section,
which is devoted to the network of Fig.~\ref{f7}(d). It turns out that the network of Fig.~\ref{f7}(c)
is a special case of the network of Fig.~\ref{f7}(d), namely, 
for $z=[x,y]$ the framework Fig.~\ref{f7}(d) essentially turns into the network of Fig.~\ref{f7}(c).


\section{Main results}
\label{s4}

 \subsection{The network of Fig.~\ref{f7}(d) and An. Muchnik's theorem} \label{muchnik} 
  
Basically this network was studied in the paper~\cite{muchnik},
\begin{figure}[t]
\begin{center}
\includegraphics{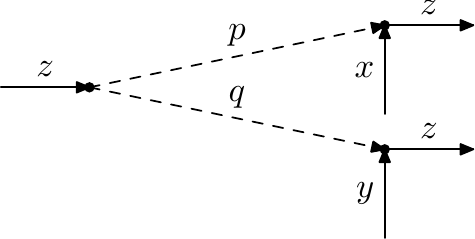}\qquad
\end{center}
\caption{The network studied in~\cite{muchnik}}\label{f6}
\end{figure}
along with the network shown on Fig.~\ref{f6} (the difference is in that the input to the left node equals $z$, and not  $(x,y,z)$). 
Using cuts we can show that  $C(z|x)$ is a lower bound for $C(p)$ and $C(z|y)$ is a lower bound for $C(q)$,
and $\max\{C(z|x),C(z|y)\}$ is a lower bound for the total disclosure (for both networks). 
These lower bounds can be attained, and moreover, they can be attained simultaneously. More specifically, 
the following holds.

\begin{theorem}[\cite{muchnik}]\label{th2}
For all $n$ 
for all strings $x,y,z$ of length at most $n$
there exist consistent strings $p,q$ (which means that one of the string
is a prefix of the other one) with
\begin{align*}
|p|&= C(z|x), \\
|q|&= C(z|y), \\
C(z|x,p)&=O(\log n), \\
C(z|y,p)&=O(\log n) ,\\ 
C(p,q|z)&=O(\log n).
\end{align*}
\end{theorem}
Note that consistency of  $p,q$ implies that  $C(p,q)=\max\{C(z|x),C(z|y)\}+O(\log n)$,
that is,  $C(p,q)$ attains its minimum.

Finally, the next theorem states the minimum private disclosure $I(p,q:x,y)$
for this network is  $\max\{ I(x:z|y), I(y:z|x)\}$ and that it is attained for the pair $p,q$
from Theorem~\ref{th2}. 

\begin{theorem}\label{th3}
(1) If $p,q$ is  $\eps$ information transmission in this circuit,
then 
$$
I(p,q:x,y)\ge \max\{ I(z:y|x), I(z:x|y)\}-O(\log n+\eps).
$$
(2) For all  $y,x,z,p,q$ satisfying the conditions and the statement of Theorem~\ref{th2}, it holds that
\begin{align*}
I(p,q:x,y)&= \max\{ I(y:z|x), I(x:z|y)\}+O(\log n).
\end{align*}
Here $n$ stands for the maximal length of $x,y,z,p,q$. 
\end{theorem}

\begin{proof}
(1) We will prove the inequality
$$
I(p,q:x,y)\ge I(z:x|y)-O(\log n)-\eps,
$$
and the other inequality can be proved in a similar way. 
This can be done using the cut consisting of the right lower node.
However this time in place of Lemma~\ref{l3} we will use Lemma~\ref{l1} (presented in the end of the proof).
By Lemma~\ref{l1} we have
$$
I(z:x|y)\le I(q:x,y)+\eps+O(\log n),
$$
therefore,
$$
I(p,q:x,y)\ge I(q:x,y)\ge I(z:x|y)-\eps-O(\log n).
$$

(2)
W.l.o.g. assume that  $C(z|x)\le C(z|y)$. Then the pair $p,q$ basically coincides with $q$.
Since 
 $$
 I(y:z|x)=C(z|x)-C(z|x,y),\quad I(x:z|y)=C(z|y)-C(z|x,y),
 $$
our assumption implies that 
   $$
   I(y:z|x)\le  I(x:z|y).
   $$
Thus we have to prove that   
\begin{align*}
I(q:x,y)&=  I(x:z|y)+O(\log n).
\end{align*} 
As  $C(z|y,q)=O(\log n)$,
by Lemma~\ref{l1} (the second part)
the sought equality means that   
$$
I(q:y)+C(z|q,x,y)+I(q:x|y,z)=O(\log n).
$$   
The first term is of order  $O(\log n)$,
since $q$ is a minimum length program for $z$ relative to $y$,
and hence has no information about $y$.\footnote{Formally,
we use here the inequality $C(z|y)\le C(q)-I(y:q)+C(z|q,y)+O(\log n)$
that holds for all strings $q,y,z$  of length at most $n$.
} 
The second term is of order
$O(\log n)$, since even the larger quantity $C(z|q,y)$ is of that order by assumption.
Finally,  $I(q:x|y,z)=O(\log n)$, since by Theorem~\ref{th2} we have $C(q|z)=O(\log n)$.
\end{proof}

\begin{lemma}\label{l1} 
For all strings $y,z,x,q$ of length at most $n$ it holds
\begin{align} \label{eq6}
I(z:x|y)\le I(q:x,y)+C(z|q,y)+O(\log n)
\end{align}
and the difference between the right hand side and the left hand side is equal to
\begin{align} 
I(q:y)+C(z|q,y,x)+I(q:x|y,z)+O(\log n).\label{eq7}
\end{align}
\end{lemma}



 \begin{proof} 
Since all the terms in the sum~\eqref{eq7} are non-negative, the second statement implies the first one.
Let us prove the second statement. Subtracting the sum~\eqref{eq7}  from the right hand side of the inequality~\eqref{eq6},
we obtain 
$$
 I(q:x|y)+ I(x:z|q,y)-I(q:x|y,z)
$$
(ignoring terms of order $O(\log n)$). Thus we have to show that
$$
 I(q:x|y)+ I(z:x|q,y)=I(z:x|y)+I(q:x|y,z).
$$
This is obvious, since by the Chain rule both hand sides are equal to
$$ 
I(q,z:x|y).
\qed$$
\renewcommand{\qed}{}\end{proof}



\subsection{The network of  Fig.~\ref{f7}(e)}

For this network minimal $C(p),C(q), C(p,q)$ are $C(y|x),C(z|x),C(y,z|x)$,
respectively. Minimal private disclosure $I(p,q:x)$ is zero, which is witnessed by the pair 
$p,q$, 
where $p=q$ is the minimum length program for  $[y,z]$ relative to $x$.

However, in general, neither the minimum $C(p)$, $C(q)$, $C(p,q)$, nor the minimum $C(p)$, $C(q)$, $I(p,q:x)$
can be attained simultaneously.
First we show that there are $x,y,z$ such that the minimum $C(p),C(q), C(p,q)$ cannot be attained simultaneously.
The gap will be one eighth of the maximal length of $x,y,z$.

\begin{theorem}\label{th10}
For all natural $\eps$ 
there are strings $x,y,z$ of length at most 
$8\eps+O(1)$ such that for all strings 
$p,q$ with
\begin{align*}
C(y|p,x)&< \eps,\\
C(z|q,x)&< \eps,\\
C(p)&< C(y|x)+\eps-O(1), \\
C(q)&< C(z|x)+\eps-O(1)
\end{align*}
(that is, the pair $(p,q)$
is $\eps$ information transmission and  complexities of $p,q$ are  close
to the minimum) it holds that 
\begin{align*}
C(p,q)&\ge C(y,z|x)+\eps -O(1)
\end{align*}
(that is, the complexity of the pair $p,q$ is $\eps$ larger than its minimum).
\end{theorem}
\begin{proof}
The lengths of $x,y,z$ will be denoted by $n,m,m$, respectively.
Besides we will denote by
$j$ the upper bound for $C(y,z|x)$.
All parameters $n,m,j$ will be linear functions of $\eps$, which will be chosen later.

Let us start an enumeration of the upper graph $\{(u,v,i)\mid C(u|v)\le i\}$ of conditional Kolmogorov complexity.
Each time a new triple in that enumeration appears, we update current 
upper bounds for 
$C(u|v)$ for all $u,v$. Those bounds will be denoted by the same letters.
Initially, $C(u|v)=+\infty$ for all $u,v$.

Observing  enumerated triples, we will define a computable function 
$A:\{0,1\}^*\times\{0,1\}^*\to\{0,1\}^*$ by enumerating its graph. 
Besides, for all natural $\eps$ dovetail style we do the following.
We call a pair $p,q$ \emph{suitable for  
$\eps$} (on a certain step of the enumeration)
if $C(p,q)<j+\eps$, $C(p)<m+\eps$, and $C(q)<m+\eps$. 
On every step of the enumeration for all $\eps$ we will define a triple of strings  $(x,y,z)$, called the  \emph{candidate triple}.
The candidate triple will always satisfy the following invariant:
\begin{itemize}
\item $C_A(y,z|x)\le j$ and 
\item there is no suitable pair
$p,q$ with $C(y|p,x),C(z|q,x)<\eps$.
\end{itemize}
Here $C_A(u|x)$ denotes complexity with respect to $A$ defined as 
$\min\{|r|\mid A(r,x)=u\}$.

Initially, we let the candidate triple to be the lex first triple (for all $\eps$). Since initially there are no suitable pairs, the 
second item of the invariant holds in a trivial way.
And to make the first item hold, we choose a string $r$ of length $j$ and let $A(r,x)=[y,z]$ (that is,
we enumerate the triple $(r,x,[y,z])$ in the graph of the function $A$).

When a new suitable pair appears (for some $\eps$), we replace in the candidate triple for that $\eps$ the first string $x$ by the lexicographically
next string. This will be always possible, since the number of available
$x$'s is $2^{n}$, which will by larger than  $2^{\max\{j,m\}+\eps}$, which is an upper bound for 
the number of times the new suitable pair can appear.
Then we choose the second and third strings   $y,z$ in the candidate triple so that the second item of the invariant 
holds. Let us see what we need to make this choice possible.
 The number of suitable pairs is less than  $2^{j+\eps}$, and for each suitable pair there are less than
 $2^{2\eps}$ pairs $y,z$ with $C(y|p,x),C(z|q,x)<\eps$. Thus the number of the pairs 
 $y,z$ that do not satisfy the second item of the invariant is less than $2^{j+3\eps}$. So we need the inequality
 $2^{j+3\eps}\le 2^{2m}$.
And again to make the first item hold, we choose a fresh string $r$ of length $j$ and let $A(r,x)=[y,z]$.

Assume now that for some suitable pair
$p,q$ the second item of the invariant becomes invalid because one of the quantities  $C(y|p,x),C(z|q,x)$ becomes less than $\eps$. 
In this case we do not change  $x$ and choose a new pair $y,z$ in the same way as in the previous paragraph. For every fixed
$x$ this can happen less than $2^{m+\eps}\cdot 2^{\eps}+2^{m+\eps}\cdot 2^{\eps}$ times (the number
of $p$'s times the number of  $y$'s for a fixed $p$,
plus the number of $q$'s times the number of $z$'s for a fixed  $q$).
We are in a good shape, if this number is less than or equal to $2^{j}$, the number of  $r$'s of length $j$ (for different 
$\eps$ the strings $x$ will be different, since they have different lengths).

For every $\eps$ there exists a step in the enumeration after which the second item of invariant 
cannot break and hence the invariant is true forever.
Thus, for the last candidate the invariant is true for the genuine Kolmogorov complexity $C$. Let us show that the invariant implies the statement of the theorem. 
Indeed, the second item guarantees that  
$C(y,z|x)\le j+O(1)$. Therefore the inequality
 $C(p,q)\ge j+\eps$ implies the inequality $C(p,q)\ge C(y,z|x)+\eps -O(1)$. Since both quantities  $C(y|x),C(z|x)$ are less than $m+O(1)$,
 the inequalities  $C(p)< C(y|x)+\eps-O(1), C(q)< C(z|x)+\eps-O(1)$ imply the inequalities
 $C(p)< m+\eps$, $C(q)< m+\eps$, respectively.

It remains to find   $n,m,j$ so that the inequalities 
 $$
 j+3\eps\le 2m,\quad  \max\{j,m\}+\eps\le n,\quad m+2\eps+1\le j
 $$
 hold.
 For example, we can let
 $$
 n=8\eps+2,\quad m=5\eps+1,\quad j=7\eps+2.
 \qed$$
\renewcommand{\qed}{}\end{proof}

\begin{remark}
We can replace $8\eps+O(1)$  by
$5\eps+O(1)$ in the statement of the theorem. To this end, we can modify the construction as follows: every time the invariant breaks,
we change the entire candidate triple in a smart way. We will use that technique in the proof of the next theorem. 
\end{remark}

We will show now that for this network, for some $x,y,z$, it is impossible to attain simultaneously 
the minimum values of $C(p),C(q), I(p,q:x)$, which are
 $C(y|x),C(z|x),0$. The gap will be about one ninth of the maximal length of
$x,y,z$.

\begin{theorem}\label{th1}
For all natural $\eps$ 
there are strings $x,y,z$ of length less than $9\eps + O(1)$ such for all  
$p,q$ with
\begin{align}\label{eqprog}
C(y|p,x)<\eps,
C(z|q,x)<\eps,\quad 
C(p)< C(y|x)+\eps-O(1),\quad 
C(q)< C(z|x)+\eps-O(1)
 \end{align}
we have
\begin{align*} 
C(p,q)-C(p,q|x)&\ge \eps -O(1)
\end{align*}
(that is, the private disclosure is  by $\eps-O(1)$ larger than its minimum). 
\end{theorem}
\begin{proof}
We will denote the lengths of $x,y,z$ by
$n,m,m$ and choose them later. Besides, we will denote by $j$ the upper bound for  $C(p,q|x)$, the value of  $j$ 
also will be chosen later.

Again we start an enumeration of the  upper graph of conditional Kolmogorov complexity.
Observing the enumeration, we define a computable function $A(u,v)$ by enumerating its 
graph. And again for every $\eps$
on every step of the enumeration we will have a candidate triple  $x,y,z$. For a candidate triple  $x,y,z$, we call $p$
\emph{a program for $y$} if
$C(y|p,x)<\eps$ and $C(p)<m+\eps$,
and we call $q$ \emph{a program for  $z$} if 
$C(z|q,x)<\eps$ and $C(q)<m+\eps$.
Here and later $C$ denotes the current upper bound for Kolmogorov complexity (on the current step
of the enumeration).

Now we keep the following invariant (for every $\eps$):
\begin{enumerate}
\item for all pairs  $p,q$ of programs for $y,z$, respectively, it holds $C(p,q)\ge j+\eps$,
\item for all pairs  $p,q$ of programs for $y,z$, respectively, it holds $C_A(p,q|x)\le j$.
\end{enumerate}
This invariant guarantees the statement of the theorem for the last candidate triple.
Indeed, since
$C(y|x), C(z|x)<m+O(1)$, all  $p,q$ satisfying the inequalities~\eqref{eqprog}
are programs for $y,z$, respectively. Therefore the invariant implies that for every such pair 
$p,q$ it holds $C(p,q)\ge j+\eps$ and $C(p,q|x)\le j+O(1)$,
hence the private disclosure is larger than $\eps-O(1)$.

Initially, the candidate triple is the lex first triple. 
Since initially there are no programs for $y,z$, the invariant is fulfilled.

In the course of the enumeration, both items of the invariant can break.
The invariant can break less than $2^{j+\eps}$ times because Kolmogorov complexity of a pair of programs $p,q$
becomes less than
 $j+\eps$ and less than 
$2^{n+m+2\eps}$ times because of appearance of a new program for $y$ or for $z$.
After any break we find a new candidate triple in such a way that the first item of the invariant holds.
Then we let $C_A(p,q|x)\le j$ for pairs $p,q$ of programs for $y,z$, respectively, thus restoring the second item.
This operation will be called  the \emph{simplifying $p,q$ with respect to $x$}.  

The new candidate triple is chosen to satisfy the statement of the following
\begin{lemma}
Assume that the total number of simplifications made so far is less than $2^{n+j-3}$. Assume further that 
$j+3\eps\le 2m-2$. Then there is a triple of strings   $x,y,z$ of lengths $n,m,m$ such that 
 \begin{itemize}
\item so far, there have been made less than  $2^{j-1}$ simplifications with respect to $x$,  
\item  $C(p,q)\ge j+\eps$ for every pair of programs $p,q$ for $y,z$, respectively, 
\item the number of programs for $y$ is less than $2^{2\eps+2}$, 
\item the number of programs for $z$ is less than $2^{2\eps+2}$.
\end{itemize}
\end{lemma}
\begin{proof}
We will show that each of these four conditions holds for more than $3/4$ of triples $x,y,z$.

Indeed, assume that the triple $x,y,z$ is chosen at random with respect to the uniform distribution.
The first assumption implies that on average there have been made less than $2^{j-3}$ simplifications with respect to a single $x$.
Hence there are less than a quarter of $x$'s with respect to which we have made more than $2^{j-1}$ simplifications.

For every pair $p,q$ and for every  $x$ there are less than $2^{\eps}$  $y$'s for which $p$ is a program
and less than $2^{\eps}$ $z$'s for which  $q$ is a program. Thus for each $p,q$ and  $x$ there are less than $2^{2\eps}$
pairs $y,z$ for which $p,q$ are programs. Hence for every $x$ there are less than
 $2^{j+\eps+2\eps}$ pairs  $y,z$ which has a pair of programs of complexity less than $j+\eps$. The second assumption implies
 that there are less than a quarter of such pairs $y,z$.
 
Finally, as we have shown in the last paragraph, for every pair  $x,p$
there are less than $2^{\eps}$  $y$'s for which $p$ is a program.
Thus for every $x$ on average a single $y$ has less than $2^{m+\eps+\eps}/2^m$
programs.
Therefore for every $x$ less than a quarter of 
 $y$'s have $2^{2\eps+2}$ (or more) programs. 
 Similar arguments work for  $z$.
 \end{proof}

Recall that the invariant can break less than $2^{j+\eps}+2^{n+m+2\eps}$ times. Hence the total 
number of simplifications is less than this number times $2^{2\eps+2}\cdot 2^{2\eps+2}$.
To find a new candidate triple we need this number be less than $2^{n+j-3}$.
Therefore we have to choose $n,m,j$  so that 
$$
(2^{j+\eps}+2^{n+m+2\eps})\cdot2^{4\eps+4}\le 2^{n+j-3}, \quad j+3\eps\le 2m-2, \quad 2^{2\eps+2}\cdot 2^{2\eps+2}\le 2^{j-1}
 $$
(we need the last inequality to be able to make  $2^{2\eps+2}\cdot 2^{2\eps+2}$ simplifications 
for the new candidate triple).
These inequalities hold for the following parameters:
$$
n=5\eps+8,\quad m=9\eps+10, \quad j=15\eps+18.\qed
$$
\renewcommand{\qed}{}
\end{proof}

\subsection{The network  Fig.~\ref{f7}(f)}

For this network the minimum  $C(p,q)$ is 
$\max\{C(y,z|x),C( z|y)\}$. The lower bound follows from Lemma~\ref{l3} applied to the following two cuts.
The cut consisting of both right nodes yields the inequality
$C(p,q)\ge C(y,z|x)-2\eps-O(\log n)$, and the cut consisting of the single bottom right node yields the inequality
 $C(q)\ge C( z|y)-\eps-O(\log n)$.

On the other hand, we can apply Theorem~\ref{th2} after replacing both output strings by 
 $[y,z]$. By this theorem there exists a pair of strings $p,q$ of complexity
$$
\max\{C(y,z|x),C(y,z|y)\}+O(\log n)
$$ 
that is 
$O(\log n)$ information transmission in this network. Since $C(y,z|y)=C(z|y)+O(1)$,
we obtain the sought upper bound.

\begin{theorem}
The minimum private disclosure  $I(p,q:x)$
for the circuit of Fig.~\ref{f7}(f)
is  $$\max\{0,C(z|y)-C(y,z|x)\}.$$
\end{theorem}
\begin{proof}
The lower bound follows Lemma~\ref{l1}. Indeed, assume that a pair  $p,q$ is 
$\eps$ information transmission in this circuit. By Lemma~\ref{l1}
we have the inequality
$$
I(z:x|y)\le I(q:x,y)+C(z|q,y).
$$ 
(All inequalities in this proof hold with accuracy
$O(\log n)$.)
By the Chain rule we have 
$$I(q:x,y)=I(q:x)+I(q:y|x)\le I(q:x)+C(y|x).
$$  
Since $C(z|q,y)<\eps$,
we obtain the inequality
$$
I(z:x|y)\le I(q:x)+C(y|x)+\eps.
$$ 
It remains to note that 
$C(z|y)-C(y,z|x)$ and $I(z:x|y)-C(y|x)$ coincide with logarithmic accuracy.

The upper bound for  $I(p,q:x)$ follows from Theorem~\ref{th2} applied to the triple
$x,y,[y,z]$.
By this theorem there is a pair $p,q$ of consistent strings that have complexities
$C(y,z|x), C(y,z|y)$, respectively, such that  $C(p,q)$ is equal to the maximum of these numbers and 
all quantities 
$$C(p,q|y,z),C(y,z|x,p),C(y,z|y,q)$$
are of order $O(\log n)$. 
Now we will distinguish two cases.

Case  1: $C(y,z|x)>C(y,z|y)$.
Then the pair $p,q$ basically equals $p$ and is independent of $x$ (as a minimum length program for  $y,z$ relative to
 $x$). That is, the private disclosure is 0.
On the other hand, the maximum is the statement of the theorem
is also 0 with accuracy  $O(1)$, since $C(y,z|y)=C(z|y)+O(1)$.

Case 2: $C(y,z|x)\le C(y,z|y)$.
Then the pair $p,q$ basically equals $q$. In this case the string $q$ may have information about $x$.
However its length-$|p|$ prefix is $p$ and hence is independent on $x$. By the Chain rule the amount of mutual 
information in $q$ and $x$ cannot exceed the difference of lengths of $p$ and $q$, which is equal to
$$C(y,z|y)-C(y,z|x)= C(z|y)-C(y,z|x)+O(1).\qed$$
\renewcommand{\qed}{}
\end{proof}

It turns out that for this network there are 
 $x,y,z$ such that $C(p),C(q), C(p,q)$ cannot simultaneously attain minimum values.

\begin{theorem}\label{th11}
For all natural  $\eps$ 
there exist strings $x,y,z$ of length at most $5\eps+O(1)$ such that 
for all  $p,q$
the inequalities 
\begin{align*}
C(y|p,x)&< \eps,\\
C(z|q,y)&< \eps,\\
C(p)&< C(y|x)+\eps-O(1), \\
C(q)&< C(z|y)+\eps-O(1)
\end{align*}
imply the inequality
\begin{align*}
C(p,q)&\ge \max\{C(y,z|x),C(z|y)\}+\eps -O(1)
\end{align*}
(that is, $C(p,q)$ is $\eps$ larger than its minimum).
\end{theorem}
\begin{proof}
The proof is similar to that of Theorem~\ref{th1}.
Unfortunately, the simpler technique from the proof of
Theorem~\ref{th10} does not help.

We will denote the lengths of $x,y,z$ by $n,m,k$  and again we will use another parameter  $j$, the upper bound for $\max\{C(y,z|x),C(z|y)\}$.
Again we enumerate the upper graph of conditional Kolmogorov
complexity and define a computable function
$A(u,v)$ by means of enumerating its graph. 
And again on every step of that enumeration for every  $\eps$ we will have a candidate triple  $x,y,z$.
For a given candidate triple $x,y,z$ we call a string $p$
a \emph{program for $y$} if $C(y|p,x)<\eps$ and $C(p)<m+\eps$.
And we call $q$ 
\emph{a program for $z$} if 
$C(z|q,y)<\eps$ and $C(q)<k+\eps$.

We maintain now the following invariant:
\begin{enumerate}
\item for every pair of programs  $p,q$ for $y,z$, respectively, it holds $C(p,q)\ge j+\eps$,
\item $C_A(y,z|x)\le j$,
\item $C_A(z|y)\le j$.
\end{enumerate}
As before this invariant guarantees the statement of the theorem for
the last candidate triple.

The second and the third items cannot break. The first item can break less than $2^{j+\eps}$ times
because Kolmogorov complexity if some pair of programs becomes less than  $j+\eps$  and 
less than $2^{n+m+2\eps}+2^{m+k+2\eps}$ times because of appearance of a new program.
After each break we find a new candidate triple satisfying the first item
of the invariant. Then we let  
$C_A(y,z|x)\le j$, which is called  
 \emph{simplifying the pair $y,z$ with respect to $x$}.  
Then we simplify $z$ with respect to $y$
 (that is, we let $C_A(z|y)\le j)$. 

The new candidate triple is chosen by means of the following 
\begin{lemma}
Assume that the total number of simplifications of both types made so far is less than  $2^{\min\{n,m\}+j-2}$.
Assume further that $j+3\eps\le m+k-2$. Then there is a triple of strings $x,y,z$ of lengths $n,m,k$, respectively,
such that
 \begin{itemize}
\item the number of simplifications of $y,z$ with respect to  $x$ made so far is less than  $2^{j}$,  
\item  the number of simplifications of $z$ with respect to  $y$ made so far is less than $2^{j}$,  
\item $C(p,q)\ge j+\eps$ for every pair of programs $p,q$ for $y,z$, respectively. 
\end{itemize}
\end{lemma}
\begin{proof}
This lemma is proven in a similar way as the similar lemma in the proof of 
Theorem~\ref{th1}. Namely, we show that each of these three conditions holds for more than three 
quarters of triples $x,y,z$.
Indeed, the first assumption implies that on average with respect to a single 
$x$ there have been made less than $2^{j-2}$ simplifications.
Hence with respect to less than a quarter of $x$'s we have made  $2^{j}$ (or more) simplifications.
A similar arguments works for simplifications with respect to $y$.

For every pair $p,q$ and for every $x$ there are less than  $2^{\eps}$ $y$'s for which  $p$ is a program,
and for each of these $y$'s there are less than $2^{\eps}$ $z$'s for which  $q$
is a program. Thus for each pair $p,q$ and for each $x$ there are less than $2^{2\eps}$
 pairs $y,z$ such that $p,q$ are programs for $y,z$, respectively. Hence for all  $x$
 there are less than  $2^{j+3\eps}$ pairs  $y,z$ such for some pair $p,q$  of complexity 
 less than  $j+\eps$ strings $p,q$  are programs for   $y,z$, respectively. By the second assumption
 there are less than a quarter of such pairs $y,z$ (for every $x$).
    \end{proof}


Recall that the invariant can break less than  $2^{j+\eps}+2^{n+m+2\eps}+2^{m+k+2\eps}$ times
and thus the number of simplifications of each type is less than this number. To meet the assumptions of the theorem we need
this number be less than  $2^{\min\{n,m\}+j-2}$.
Thus it suffices to choose $n,m,j$ so that 
$$
2^{j+\eps}+2^{n+m+2\eps}+2^{m+k+2\eps}\le 2^{\min\{n,m\}+j-2}, \quad j+3\eps\le m+k-2
 $$
These inequalities are satisfied, for instance, by the following values:
$$
n=m=k=5\eps+6,\quad j=7\eps+10.\qed
$$
\renewcommand{\qed}{}
\end{proof}

The next theorem claims that for some $x,y,z$ 
it is impossible to attain the minimum of the quantities 
$C(p),C(q),I(p,q:x)$ simultaneously. 

\begin{theorem}\label{th4}
For all natural $\eps,l$ 
there are strings  $x,y,z$ of length at most $9\eps+l+O(1)$ such that  
for all strings $p,q$ with
\begin{align*}
&C(y|p,x)<\eps,\quad
C(z|q,y)<\eps,&
&C(p)< C(y|x)+\eps-O(1), \quad
C(q)< C(z|y)+\eps-O(1)&
 \end{align*}
we have that 
\begin{align*} 
C(p,q)-C(p,q|x)&\ge l+\eps -O(\log l),\qquad C(z|y)-C(y,z|x)\le l+O(1) 
\end{align*}
(that is, the private disclosure is about $\eps$ larger than its minimum, which is less than the given number $l$).
\end{theorem}
\begin{proof}
The proof is similar to that of Theorem~\ref{th1} but 
it is more complicated because we have to
ensure the upper bound for $C(z|y)-C(y,z|x)$ and because  $y$ and $z$
are not symmetric any more.
As before we will denote the lengths of $x,y,z$ by 
$n,m,k$, which will be chosen later. Besides we will need the anticipated upper bound for $C(p,q|x)$, which will be denoted by 
$j-l$ and will be chosen later as well.

Again we enumerate the upper graph of conditional Kolmogorov 
complexity and define a computable function
$A(u,v)$ by enumerating its graph. 
Again at each step of this enumeration we will have a candidate triple $x,y,z$. For a fixed candidate triple 
 $x,y,z$ we call $p$
\emph{a program for $y$} if  $C(y|p,x)<\eps$ and $C(p)<m+\eps$,
and we call $q$ \emph{a program for $z$} if  $C(z|q,y)<\eps$ and $C(q)<k+\eps$.

We keep the following invariant:
\begin{enumerate}
\item $C(y,z|x)\ge k-l$,
\item for every pair $p,q$ of programs for $y,z$, respectively, it holds that  $C(p,q)\ge j+\eps$,
\item for every pair $p,q$ of programs for $y,z$, respectively, it holds that $C_A(p,q|x,l)\le j-l$.
\end{enumerate}
As before, the invariant guarantees the statement of the theorem 
for the last candidate. 
And again, initially the candidate triple is the lex first triple and each time
the invariant breaks we change the candidate triple.
The invariant can break by any of the following reasons:
\begin{itemize}
\item The inequality $C(y,z|x)< k-l$ becomes true. This can happen less than  $2^{n+k-l}$ times.
\item A new pair $p,q$ with $C(p,q)< j+\eps$ appears.  This can happen less than  $2^{j+\eps}$ times.
\item A new $p$ with $C(p)< m+\eps$
or a new $q$ with $C(q)< k+\eps$ appears. This can happen less than  $2^{\max\{m,k\}+\eps}$ times.
\item For some  $p$ with $C(p)< m+\eps$ the inequality 
$C(y|p,x)<\eps$ becomes true or for some  $q$ with $C(q)< k+\eps$ the inequality
$C(z|q,y)<\eps$ becomes true. This can happen less than $2^{n+m+2\eps}+2^{m+k+2\eps}$ times.
\end{itemize}
After each change of the candidate triple 
we let  $C_A(p,q|x,l)\le j-l$ for all pairs $p,q$, where $p,q$ are programs for
$y,z$, respectively.  
We call this operation \emph{simplifying  $p,q$ with respect to $x$}.  The new triple will be chosen so that
the number of such pairs $p,q$ be less than $2^{j-l-1}$.
More specifically,  the new candidate triple  $x,y,z$ so that the following hold:
 \begin{itemize}
\item the total number of simplifications with respect to 
$x$ made so far is less than $2^{j-l-1}$,  
\item $C(y,z|x)\ge k-l$, 
\item for every pair of programs $p,q$ for  $y,z$, respectively, it holds $C(p,q)\ge j+\eps$, 
\item the number of programs for  $y$ is less than $2^{2\eps+3}$, 
\item the number of programs for   $z$ is less than $2^{2\eps+3}$.
\end{itemize}

Assume that the total number of simplifications made so far is less than
$2^{n+j-l-4}$. Assume further that $\max\{k-l,j+3\eps\}\le m+k-3$. 
Then each of these five conditions holds for more than $7/8$ of all triples $x,y,z$.

Indeed, the first assumption implies that on average, with respect to one $x$ less than  $2^{j-l-4}$ simplifications have been made.
Thus there are less than one eighth of $x$'s with respect to which $2^{j-l-1}$ (or more) simplifications have been made.

Similarly, the second assumption implies that for every $x$ there are less than one eighth of pairs  $y,z$ with $C(y,z|x)< k-l$. 

For every pair  $p,q$ and for every $x$ there are less than  $2^{\eps}$ $y$'s for which $p$
is a program. For every of such $y$'s there are less than $2^{\eps}$  $z$'s for which $q$
is a program. Therefore, for every pair $p,q$ and for every  $x$ there are less than $2^{2\eps}$
 pairs $y,z$ such that  $p,q$ are programs for $y,z$, respectively. Hence for every $x$
there are less than $2^{j+\eps+2\eps}$ pairs $y,z$ such that there is a pair $p,q$ of complexity less than 
$j+\eps$ where $p$ is a program for $y$ and $q$ is a program for $z$. The second assumption
implies that there are less than one eighth of such pairs $y,z$.
 
Finally, as shown in the last paragraph, for every pair
 $x,p$ there are less than 
$2^{\eps}$  $y$'s for which $p$ is a program. Therefore, for every $x$
on average one $y$ has less than  $2^{m+\eps+\eps}/2^m$ programs.
Hence for every 
$x$ less than one eighth of 
 $y$'s have $2^{2\eps+3}$ (or more) programs.
Similarly, for every pair  $x,y$ less than one eighth of 
$z$'s have $2^{2\eps+3}$  (or more) programs.
 
Thus more than 3/8 of triples $x,y,z$ satisfy all five conditions.
We choose one of them and simplify
all pairs $p,q$ where $p,q$ are programs for $y,z$, respectively.  
These simplifications are indeed possible provided  $2^{2\eps+3}\cdot 2^{2\eps+3}\le 2^{j-l-1}$.

Now we have to show by induction that each time the invariant breaks
the assumptions that ensure the existence of a new candidate triple are met.
The total number of changes of the candidate triple is less than 
 $$
2^{n+k-l}+2^{j+\eps}+2^{\max\{m,k\}+\eps}+2^{n+m+2\eps}+2^{m+k+2\eps}<2^{n+k-l}+2^{j+\eps}+2^{n+m+2\eps}+2^{m+k+2\eps+1}.
 $$
To satisfy the upper bound for the total number of simplifications
it suffices to ensure that all four terms in the right hand side of this inequality 
multiplied by  $2^{4\eps+6}$ be less than $2^{n+j-l-6}$.
Besides we used the inequalities $\max\{k-l,j+3\eps\}\le m+k-3$
and $2^{4\eps+6}\le 2^{j-l-1}$.
So it suffices to make true the following inequalities:
\begin{align*}
\max\{n+k-l, j+\eps,n+m+2\eps,m+k+2\eps+1\}&\le n+j-l-4\eps-12,\\
\max\{k-l,j+3\eps\}&\le m+k-3\\
4\eps+6&\le j-l-1
\end{align*}
It is easy to verify that the following values of parameters satisfy these inequalities:
$$
n=k=9\eps+l+16,\quad m =  7\eps+15, \quad j=13\eps+l+28.\qed
$$
\renewcommand{\qed}{}
\end{proof}

\section{Open questions}

1) It is interesting to investigate  the case when
the eavesdropper has some a priori information which is represented by a string  $s$.
For the network of  Fig.~\ref{f7}(b), this question was studied  in~\cite{much} for the case when the eavesdropper needs
information about the string $y$. More specifically, in~\cite{much} for given $x,y,s$ the minimum $I(p:y|s)=C(y|s)-C(y|p,s)$ over
all $p$'s with $U(p,x)=y$ was considered (where $U$ stands the optimal programming language).
Obviously $I(p:y|s)$ cannot be less than $C(y|s)-C(x|s)-O(1)$.
In~\cite{much}, it was shown that this minimum is equal to  $\max\{C(y|s)-C(x|s),0\}+O(\log n)$ where $n=\max\{|x|,|y|,|s|\}$. 
However the minimum value was attained for $p$ of length exponential in $n$ and it remains open whether 
it can be attained for $p$ of polynomial length. Besides, in  \cite{much} it was proved that
the minimum can be attained for $p$ of minimal possible  length $C(y|x)+O(\log n)$ in the case when $C(x|s)\ge C(y|s)+C(y|x)+O(\log n)$ 
(in this case the minimum is of order $O(\log n)$).

2) For the network of  Fig.~\ref{f7}(b), it is interesting to study even a more general problem: the eavesdropper 
has some a priori information $s$ and needs 
information about a fourth string $t$ (that can differ from 
$y$,  $x$ and $[x,y]$).

\end{document}